\documentclass[preprint,12pt]{elsarticle}

\usepackage{amsfonts,amsmath,amssymb,amsthm,mathrsfs,color}
\usepackage[a4paper]{geometry} % to change the page dimensions
\usepackage{graphicx} % support the \includegraphics command and options
\usepackage[colorlinks]{hyperref}

\usepackage{times}

\newtheorem{theorem}{Theorem}[section]

\newtheorem{lemma}[theorem]{Lemma}
\newtheorem{proposition}[theorem]{Proposition}

\newcommand{\E}{\mathbb{E}^{\mathbb{Q}}}

\begin{document}

\begin{frontmatter}

\title{\large{\bf Analytic properties of American option prices under a modified Black--Scholes equation with spatial fractional derivatives}}


\author[jiangnan]{Wenting Chen}
\author[uow]{Kai Du}
\author[uow]{Xinzi Qiu}
\address[jiangnan]{School of Business, Jiangnan University, Wuxi, Jiangsu Province, 214100, China.}
\address[uow]{School of Mathematics and Applied Statistics, University of Wollongong NSW 2522, Australia.}

\begin{abstract}
This paper investigates analytic properties of American option prices under the finite moment log-stable (FMLS) model.
Under this model the price of American options is characterised by the free boundary problem of a fractional partial differential equation (FPDE) system.
Using the technique of approximation we prove that the American put price under the FMLS model is convex with respect the underlying price, and specify the impact of the tail index on option prices.
\end{abstract}


\end{frontmatter}


\section{Introduction}

The celebrated Black--Scholes (BS) model is based on the assumption that the underlying asset price follows a geometric Brownian motion \cite{black1973pricing}. 
However, it is well documented in the literature that the BS model usually underestimates the probability of the appearance of jumps or large movements of stock prices over small time steps~\cite{CW};
for example, when analyzing the S\&P 500 data, a ``leptokurtic distribution'' is observed, which has a higher peak and two heavier tails than those of the normal distribution.
Many efforts have been made to develop mathematical models that capable of capturing the leptokurtic feature observed in financial market data.
A feasible approach is to adopt a L\'evy process extending Brownian motion for the description of the price, such as the Press model~\cite{press1967compound}, Merton's jump diffusion model~\cite{Mer}, and so on.
This idea allows us to model large price changes due to sudden exogenous events or information, and can explain some systematic empirical biases with respect to the BS model.

Among all the L\'evy process models, the finite moment log-stable (FMLS) process model, proposed by Carr and Wu~\cite{CW}, can not only successfully capture the high-frequency empirical probability distribution of the S\&P 500 data, 
but also fit simultaneously volatility smirks at different maturities. 
Most importantly, in contrast to many other L\'evy process models, 
the FMLS model ensures the finiteness of all moments of the underlying index level and the existence of an equivalent martingale measure. 
This paper is carried out under the framework of the FMLS model. 
The extension of our approach and results to other L\'evy processes models (e.g., KoBol and CGMY mentioned in~\cite{CCN}) is quite promising.

Mathematically, the key tool to characterize the non-locality induced by the pure jumps under the FMLS model is the fractional partial differential equation (FPDE), which is a subset of the class of pseudo-differential equations. 
In the new FPDE associated with the FMLS model, the second-order spatial derivative involved in the standard BS equation is replaced by an $\alpha$-order spatial derivative with $\alpha\in(1, 2]$. 
In comparison to the derivative of integer order, the fractional-order derivative at a point not only involves properties of the function at that particular point, but also the information of the function in a certain subset of the entire domain of definition. 
This global dependency of the fractional derivative has added an additional degree of difficulty when either analytical methods or the numerical simulations are attempted. 

Under the FMLS model, many techniques have been developed to compute option values, as summarized in \cite{schmelzle2010option}.
Cartea and del-Castillo-Negrete~\cite{CCN} considered the pricing of barrier options under the FMLS model purely numerically, by using a finite difference method.
Recently, Chen et al.~\cite{CXZ1} derived closed-form analytical solutions for European-style options under the same model, which is one of the main tools in our approach.
In contrast, pricing American options is more complicated even in the classical BS model~\cite{huang1996pricing,ju1998pricing,longstaff2001valuing,Zhu}, 
with challenge mainly stemming from the nonlinearity originated from the inherent characteristic that an American option can be exercised at any time during its lifespan. 
This additional right casts the American option pricing problem into a free boundary problem, which is highly nonlinear and far more difficult to deal with.
Several numerical scheme for pricing American options were proposed in \cite{CXZ,chen2014penalty,chen2015finite} under the FMLS model.

The purpose of this paper is to investigate analytic properties of American options under the FMLS model.
This topic has been extensively studied in the literature under the standard BS model. 
For example, Ekstr\"om \cite{EE1} proved that the price of an American put is convex with respect the underlying price  (see also El Karoui et al. \cite{El} and Hobson \cite{Hob} for related results),
and convexity of the optimal exercise boundary for an American put was obtained by Ekstr\"om \cite{EE2} and Chen et al.~\cite{CCJZ}.
Those properties can help us understand the asymptotic
behavior of the optimal exercise boundary near expiry \cite{CCJZ}. 
The first main result of this paper (see Theorem~\ref{convex} below) is an extension of the above results, which shows that the American put price under the FMLS model is also convex with respect the underlying price.
Regarding the difference between the FMLS model and the BS model, our second main result (see Theorem~\ref{mt1} below) proves that the price of an American put is monotonically decreasing with respect to the index $\alpha$ in the FMLS model.
This phenomenon was first observed in~\cite{CXZ1, CXZ} from their numerical examples, for which we provide a rigorous analytic proof in this paper.
As the index $\alpha$ measures the fatness of the tail of the return distribution and the FMLS model reduces to the BS model when $\alpha\uparrow 2$, our second result actually reveals that the BS model tends to underprice put options, 
and the pricing bias becomes larger as the tail becomes fatter, which corresponds to smaller $\alpha$ values.

This paper is organised as follows. 
In Section \ref{sec2} we introduce the FMLS model and present the mathematical formulation of pricing American options under this model.  
In Section \ref{sec3}  we approximate an American put by a sequence of Bermudan options.
A Bermudan option is an American-style option with a restricted set of possible exercise dates, being a combination of American and European options.
The approximation result provides us with another main tool to prove our main theorems in the next two sections.
In Section \ref{sec4} we obtain the convexity results for American options under the FMLS models,
and in Section \ref{sec5} we discuss the influence of the tail index on the option prices. 
Concluding remarks are given in the last section.

\section{American options under FMLS models}\label{sec2}

An American option is an instrument which gives the owner the right to buy/sell one asset of a certain stock at a fixed price at any time prior to some pre-determined expiration time $T$.
In this paper we consider a market consisting of a risk-free asset with constant return $r\ge 0$, and a risky asset with risk-neutral price process $(S_t)_{t\ge 0}$.
Under the risk neutral measure $\mathbb{Q}$, the price of an American option is given as a function of the current time $t\leq T$ and the current underlying price $S_t = S>0$ by
\begin{equation}\label{price}
	V(S,t) = \sup_{\gamma \in [t,T]} \mathbb{E}^{\mathbb{Q}} 
	\big[e^{-	r(\gamma-t)} g(S_\gamma) \big],
\end{equation}
where $g(\cdot)$ is the \emph{pay-off} function and $\gamma$ is a stopping time with respect to the given filtration. 
As the American call on an asset with no dividends usually has the same value of the corresponding European call~\cite{schreve2004stochastic}, in this paper we focus on the American put for which the pay-off function has the form
\begin{equation}\label{put}
g(S) = (K-S)^+ = \max\{K-S, 0\}
\end{equation}
with a given \emph{strike price} $K>0$.
 
The FMLS model assumes that, under the risk neutral measure $\mathbb{Q}$,  the log value of the underlying, i.e., $x_t = \ln S_t$, follows a stochastic differential equation of the maximally skewed LS process:
	\begin{equation}\label{FMLS SDE}
		d x_t = (r-\nu) dt +\sigma dL_t^{\alpha,-1},
	\end{equation}
where $\nu=\sigma^\alpha\sec\frac{\alpha\pi}{2}$ is a convexity adjustment \cite{CXZ}.
In another words, the FMLS model adopts the L\'evy process $L_t^{\alpha,-1}$ instead of the Brownian motion in the standard BS model.

In general, $L_t^{\alpha,\beta}$ denotes the L\'evy-Stable (LS) process with $\alpha\in(0,2]$ being the tail index describing the deviation from Brownian motion and $\beta\in[-1,1]$ being the skew parameter. 
The tail index $\alpha$ is usually restricted to $(1,2]$ such that the underlying return has the support on the whole real line~\cite{CW}.
In the maximally skewed LS process, namely $\beta=-1$, the random variable $x_t$ is maximally skewed to the left, meaning that the right tail of the distribution is fast decaying so that exponential moments exit. 
When $\alpha=2$ the FMLS model becomes the BS model, while when $\alpha<2$ the situation is much more complicated since fractional partial differential equations (FPDE) and non-local operators are involved. 
Market observations show that $\alpha$ is usually around $1.4$ \cite{CW}. 

With the change of variables $x=\ln S$, the American put price $V$ as a function of $x$ and $t$ (we will write $V=V(x,t;\alpha)$ if no confusion occurs) satisfies the free boundary problem 
	\begin{equation}\label{FMLS}
		\frac{\partial V}{\partial t} + (r+\frac12\nu)\frac{\partial V}{\partial x} - \frac12\nu D_x^\alpha V - rV = 0 \quad\mbox{ if }x>x(t;\alpha),
\end{equation}
where $\{x=x(t;\alpha)\}$ is the logarithm of the optimal exercise boundary, and $D_x^\alpha$ is the fractional differential operator of order $\alpha$, interpreted in the Caputo sense, namely for a $C^2$ function $f$,
	\begin{equation*}
		D_x^\alpha f(x) := \frac{1}{\Gamma(2-\alpha)}\int_{-\infty}^x \frac{f''(y)}{(x-y)^{\alpha-1}}dy, \quad 1<\alpha<2.
	\end{equation*}
From the above definition one can observe that the fractional differentiation is non-local, and it involves the option price in the exercise region $(-\infty,x(t)]$.	
Along with \eqref{FMLS}, we have the far-field boundary condition and the terminal condition for the American put price:
\begin{equation}\label{terminal}
	\begin{aligned}
	\lim_{x\to\infty} V(x,t;\alpha)&=0, \\
	V(x,T;\alpha)&=(K-e^x)^+,
	\end{aligned}
	\end{equation}
In order to properly close the FPDE system, we impose the following two moving boundary conditions:
\begin{equation}\label{bound}
	\begin{aligned}
		V(x,t;\alpha) = K - e^x &\qquad\mbox{ if }x=x(t,\alpha)\\
		\frac{\partial V}{\partial x}(x,t;\alpha) = -e^x &\qquad\mbox{ if }x=x(t,\alpha).
	\end{aligned}
	\end{equation}
	
	The above system \eqref{FMLS}--\eqref{bound} is the free boundary problem satisfied by the American put price $V(t,x;\alpha)$ and the optimal exercise boundary $x(t,\alpha)$ under the FMLS model. 
One can see that when $\alpha\to2$ it approaches to the classical BS equation.

\section{Approximating an American option with Bermuda options}	\label{sec3}

The \emph{Bermudan} option is an option where the buyer has the right to exercise at a set (always discretely spaced) number of times,
which is intermediate between a European option and an American option.

Let $g=g(\cdot)$ be a pay-off function (later on, $g(S)$ is chosen to be $(K-S)^+$ with a strike price $K$ in our cases). 
With predetermined times $0=t_0<t_1<\cdots<t_M=T$ 
the price of a Bermudan option is
	\begin{equation}\label{Bermudan option}
	B(S,0) := \sup_{\gamma\in\{t_0,t_1,\cdots,t_M\}} \E [e^{-r\gamma} g(S_\gamma)],
	\end{equation}
where $\gamma$ is a stopping time.
Using the dynamic programming principle, the family of prices can be calculated as follows:
\begin{itemize}
\item[(1)] The price $B(S,t_M)$ at $t_M=T$ is $g(S)$;
\item[(2)] Given the price $B(\cdot,t_i)$, for $1\leq i\leq M$, the price at $t_{i-1}$ is
	\begin{equation*}
	B(S,t_{i-1})=\max\left\{\E_{s,t_{i-1}}[e^{-r(t_i-t_{i-1})}B(S_{t_i},t_i)],\ g(S)\right\}.
	\end{equation*}
\end{itemize}

From the above induction, one can see that the Bermudan option is a kind of link from European options to American options. 
In fact, the price $B(S,t_{m-1})$ of a Bermudan option at $t=t_{m-1}$ can be calculated inductively as the maximum of $g(S)$ and the price of a European option with expiry $t_m$ and contract function $B(S,t_m)$.
On the other hand, we can show in the next lemma that the price of an American option can be approximated by a sequence of Bermudan option. 
To this end, we let
\begin{equation*}
		A_N := \{ 0,T2^{-N}, 2T2^{-N}, 3T2^{-N},\cdots,T\},
	\end{equation*}
and
	\begin{equation}\label{Bermudan sequence}
	B_N(S,0)=\sup_{\gamma\in A_N} \E [e^{-r\gamma}g(S_\gamma)].
	\end{equation}

\begin{lemma}\label{lec3}
Assume that $g$ is a non-negative function. As the possible exercise times of the Bermudan option gets denser, the Bermudan option price converges to the corresponding American option price under  the FMLS  models, i.e.
	\begin{equation*}
		B_N(S,0) \to V(S,0),
	\end{equation*}
	as $N\to\infty$, where $V(S,0)=\sup_{0\leq\gamma\leq T}\mathbb{E}[e^{-r\gamma}g(S_\gamma)]$ is the American option price. 
\end{lemma}
	
\begin{proof}
Under the standard Black-Scholes model, this convergence was proved in \cite{EE1}.
Here, our proof follows from there by using the dominated convergence theorem. 
Given a stopping time $\gamma\in[0,T]$, let 
	\begin{equation*}
		\gamma_N : =\inf\{ \tau\geq\gamma : \tau\in A_N\}.
	\end{equation*}
One can see that $\gamma_N\in A_N$ for any $N$, $\gamma_N\to\gamma$ almost surely as $N\to\infty$.
By the dominated convergence theorem,
	\begin{equation*}
	|\E [e^{-r\gamma}g(S_\gamma)]-\E [e^{-r\gamma_N}g(S_{\gamma_N})]| \leq \E|e^{-r\gamma}g(S_\gamma)-e^{-r\gamma_N}g(S_{\gamma_N})|\to 0,
	\end{equation*}
as $N\to\infty$. Since $g\geq0$, the above inequality implies that 
	\begin{equation*}
		{\lim\inf}_{N\to\infty}B_N(S,0)\geq V(S,0).
	\end{equation*}
On the other hand, by definitions \eqref{Bermudan option} of Bermudan option prices, $B_N(S,0)\leq V(S,0)$ for all $N$. Therefore, the lemma is proved. 
\end{proof}

\section{Convexity of option prices}\label{sec4}

In this section we aim to prove that the American put price $V(S,t)$ is convex with respect to the current underlying $S$ under the FMLS model.
Our starting point is the explicit closed-form analytical solution for European options under the FMLS model obtained by the recent work \cite{CXZ1}. 
With the pay-off function $g(\cdot)$ the price of a European option is given by \[V_E(S,t) = \E [e^{-r(T-t)} g(S_T)].\]
Define
	\begin{equation*}
		\tau = \nu(T-t),
	\end{equation*}
where for simplifying the computation, we assign $\nu=-\frac12\sigma^\alpha\sec\frac{\alpha\pi}{2}$ with a negative sign to that of the convexity adjustment in \eqref{FMLS SDE}. This transformation changes the backward problem into a forward problem. 
Also define $\gamma = r/\nu$ as the relative interest rate of the volatility with fractional order $\alpha$ to the risk-free interest rate $r$. 
One has the explicit analytical expression
	\begin{equation}\label{VF}
		V_E(x,\tau;\alpha) = \int_{-\infty}^{\infty} e^{-\gamma\tau}\Pi(x-(1-\gamma)\tau-\tau^{\frac{1}{\alpha}}m)f_{\alpha,0}(|m|)dm,
	\end{equation} 
where  $x=\ln S$ and
\begin{equation}\label{put2}
\Pi(x)=(K-e^x)^+
\end{equation} 
is the pay-off function corresponding to \eqref{put}, and $f_{\alpha,0}$ is the L\'evy stable density given by	
\begin{equation*}
		f_{\alpha,0}\left(\frac{|z|}{\tau^{1 / \alpha}}\right) = \frac{1}{\alpha}H^{1,1}_{2,2}\left[\frac{|z|}{\tau^{1 / \alpha}}\left|
		\begin{array}{cc}
		(1-\frac{1}{\alpha}, \frac{1}{\alpha}) & (\frac12,\frac12) \\
		(0,1) & (\frac12, \frac12)
		\end{array}
		\right.\right],
	\end{equation*} 
where $H(x)$ is the Fox function \cite{CXZ1}, whose general forms are defined by
	\begin{equation}\label{Hfunc}
	\begin{split}
		& H^{m,n}_{p,q}\left[ z \left|\begin{array}{cccc}(a_1,A_1) & (a_2,A_2) & \cdots & (a_p,A_p) \\ (b_1,B_1) & (b_2,B_2) & \cdots & (b_q,B_q) 	\end{array}\right.\right] \\
		=&\ \frac{1}{2\pi i}\int_L \frac{\left(\Pi_{j=1}^m\Gamma(b_j+B_js)\right)\left(\Pi_{j=1}^n\Gamma(1-a_j-A_js)\right)}{\left(\Pi_{j=m+1}^q\Gamma(1-b_j-B_js)\right)\left(\Pi_{j=n+1}^p\Gamma(a_j+A_js)\right)} z^{-s}\,ds
	\end{split}
	\end{equation}
where $L$ is a certain contour separating the poles of the two factors in the numerator. 
For detailed computation via Fourier transform and technique issues we refer the reader to \cite{CXZ1}.
But we would like to point out that the dependency of $x$ in the integrand of \eqref{VF} is only on the pay-off function $\Pi$.

Once having the above closed form analytical solutions for European options, we can derive their convexities with respect to various parameters. Then via the Bermudan option, the convexity can pass over to American options correspondingly. 

\begin{lemma}\label{lec1}
With the pay-off function $g$ given in \eqref{put}, the European option price $V_E$ is convex in the current underlying $S$.
\end{lemma}
	
\begin{proof}
Recalling that $S=e^x$, the function $\Pi(x)=g(S)=(K-S)^+$ is convex with respect to $S$, where $g(\cdot)$ is defined in \eqref{put}.
In the integrand of the above analytical expression \eqref{VF} for $V_E$ under the FMLS model, the only part involving the underlying $S$ is 
	\begin{equation*}
		\Pi(x-(1-\beta)\tau-\tau^{\frac{1}{\alpha}}m) = (K-C(m)S)^+,	
	\end{equation*}
where $C(m)=e^{-(1-\beta)\tau-\tau^{\frac{1}{\alpha}}m}$ is a positive factor for all $m\in\mathbb{R}$. 
Therefore, by integration $V_E$ is convex in $S$.  
\end{proof}	

Next, let's pass the above convexity over to the Bermudan option price $B(s,t)$ defined in \eqref{Bermudan option}. 
Since the Bermudan option price is inductively constructed by the maximum of European option prices, the following lemma is immediate.
\begin{lemma}\label{lec2}
With the pay-off function given in \eqref{put}, the Bermudan option price $B(S,t)$ is convex in the underlying $S$ for any fixed $t$. 
\end{lemma}

\begin{proof}
It is well known in convex analysis that the supremum of a family of convex functions is still convex. 
Under the FMLS model, the Bermudan option price $B(S,t)$ in \eqref{Bermudan option} is inductively constructed by the maximum of European option prices $V_E$. 
From Lemma \ref{lec1} we know that $V_E$ is convex in $S$. 
Hence for any fixed $t$ taking the supremum it is readily to see that $B(S,t)$ is convex with respect to $S$, as inherited from $V_E$. 
\end{proof}
	
We are now in a position to prove the first main result in this paper.
\begin{theorem}\label{convex}
The American option price $V(S,t)$ given by \eqref{price} with the pay-off function $g(S)=(K-S)^+$ is convex in the underlying $S$ for each $t$.
\end{theorem}
	
\begin{proof}	
To prove the convexity of $V=V(S,t)$ with respect to $S$, one can fix a time $t$ with $t\leq T$. 
Without loss of generality we may regard $t=0$, otherwise make a change $\tau \to \tau-t$. 
From Lemma \ref{lec3} we know the approximating sequence $\{B_N(S,0)\}$ converges to $V(S,0)$ as $N\to\infty$. 
And from Lemma \ref{lec2}, for each $N$ the Bermudan option price $B_N(S,0)$ is convex in $S$.
Since the point-wise limit of a convergent sequence of convex functions is again convex, we have the limit $V(S,0)$ is convex in $S$.
Therefore, the theorem is proved. 
\end{proof}

\section{The impact of fat tails on option prices}\label{sec5}

One of the major advantages of the FMLS model over the standard BS model is that it captures the fat tail feature (leptokurtosis) observed in real world markets \cite{CW}. 
Indeed, compared to the Gaussian density of the underlying prices under the BS model, the L\'evy density increases the probability of the stock price exhibiting large moments or jumps over small time steps, and has fatter tails at both ends than the lognormal distribution of the BS model.

How the fat tails impact on European option prices under the FMLS model has been illustrated in \cite{CXZ1} by their numerical experiments, but without rigorous proof.
They observed that once the tail index $\alpha$ increases up to $2$, the option prices are gradually decreasing to the BS price. 
In other words, the BS formula tends to underprice European puts with underlying following a L\'evy process. 
Moreover, the pricing bias of the BS formula gets larger as $\alpha$ becomes smaller.
They also gave a plausible explanation from a financial point of view for this phenomenon. 

In the following we give a rigorous analytical proof for the above observation  for European options and also for American options.

\begin{proposition}
As $\alpha$ approaches to $2$, the price of a European put $V_{E}$ computed in \eqref{VF} decreasingly converges to its price determined under the standard BS model in a sufficiently large underlying region $x\geq x_0$. 
\end{proposition}

\begin{proof}
Recall the analytical expression for the European put option price in \eqref{VF}, or equivalently 
	\begin{equation}\label{FMLS VFp}
		V_{E}(x,\tau;\alpha) = Ke^{-\gamma\tau}\int_{d_1}^\infty f_{\alpha,0}(|m|)dm - e^x\int_{d_1}^\infty e^{-\tau-\tau^{\frac{1}{\alpha}m}}f_{\alpha,0}(|m|)dm,
	\end{equation}
where $d_1=\frac{x-\ln K - (1-\gamma)\tau}{\tau^{\frac{1}{\alpha}}}$. 
It suffices to show that $\frac{\partial V_{E}}{\partial \alpha}<0$, as $\alpha\to2$ and $x\geq x_0$ sufficiently large. 
Since \eqref{FMLS VFp} is quite involved, we need to compute derivatives for each function of $\alpha$ separately. 

Adopting the notations from \cite{CXZ1},
	\begin{equation*}
		\nu=-\frac{1}{2}\sigma^\alpha\sec\frac{\alpha\pi}{2}
	\end{equation*}
is the convexity adjustment. By differentiating, 
	\begin{equation*}
		\frac{\partial\nu}{\partial\alpha}=\left(-\frac12\sec\frac{\alpha\pi}{2}\right)\frac{\partial\sigma^\alpha}{\partial\alpha} - \frac12\sigma^\alpha\frac{\partial\sec\frac{\alpha\pi}{2}}{\partial\alpha},
	\end{equation*}
and 
	\begin{equation*}
		\frac{\partial\sigma^\alpha}{\partial\alpha} = \frac{\partial e^{\alpha\ln\sigma}}{\partial\alpha} 
			= \sigma^\alpha\left(\ln\sigma + \frac{\alpha}{\sigma}\frac{\partial\sigma}{\partial\alpha} \right).
	\end{equation*}
From the normalisation of volatilities to $\sigma_{BS}=0.25$, one can see that $\partial\sigma/\partial\alpha\to0$ and $\ln\sigma\sim-1.38$, as $\alpha\to2$. Hence, one has $\partial\sigma^\alpha/\partial\alpha<0$, and furthermore as $\alpha\to2$, $\sec\frac{\alpha\pi}{2}\to-1$,
	\begin{equation}\label{partial1}
		\frac{\partial\nu}{\partial\alpha} < -\frac{\pi}{4}\sigma^\alpha\left(\frac{\sin\frac{\alpha\pi}{2}}{\cos^2\frac{\alpha\pi}{2}}\right)<0.
	\end{equation}

The relative interest rate $\gamma$ and the backward time $\tau$ are defined by 
	\begin{equation*}
		\gamma=\frac{-2r}{\sigma^\alpha\sec(\alpha\pi/2)}=\frac{r}{\nu},
	\end{equation*} 
and 
	\begin{equation*}
		\tau=-\frac12\sigma^\alpha(\sec\frac{\alpha\pi}{2})(T-t)=\nu(T-t),
	\end{equation*} 
respectively, where $r$ is the risk free interest.
It is readily to see that $\gamma\tau=r(T-t)$ is independent of $\alpha$.
	
The L\'evy density satisfies the inverse power-law asymptotically at large underlying values, thus for sufficiently large $x$,
	\begin{equation}\label{partial2}
		f_{\alpha,0}(x) \sim \frac{1}{|x|^{1+\alpha}},\quad\mbox{ and }\quad \frac{\partial f_{\alpha,0}}{\partial \alpha} <0.
	\end{equation}

Next we compute the partial derivative of $d_1$ in $\alpha$. Note that $\gamma\tau$ is independent of $\alpha$, we have
	\begin{equation*}
	\begin{split}
		\frac{\partial d_1}{\partial \alpha} &= (x-\ln K-(1-\gamma)\tau)\left[\frac{1}{\tau^{\frac{1}{\alpha}}}\left(-\frac{\partial\tau}{\partial\alpha}\right) + \frac{\partial \tau^{-\frac{1}{\alpha}}}{\partial\alpha} \right]\\
		&=:(x-\ln K-(1-\gamma)\tau)(I_1+I_2),
	\end{split}
	\end{equation*}
where the coefficient $(x-\ln K-(1-\gamma)\tau)>0$ when $x$ is sufficiently large. From \eqref{partial1},
	\begin{equation*}
		I_1=\frac{1}{\tau^{\frac{1}{\alpha}}}\left(-\frac{\partial\tau}{\partial\alpha}\right)=\tau^{-\frac{1}{\alpha}}(T-t)\left(-\frac{\partial\nu}{\partial\alpha}\right)>0;
	\end{equation*}
and noting that as $\tau\to0$, $\tau\ln\tau\to0$, which means the speed of $\frac{1}{\tau}$ goes to infinity is much faster than that of $|\ln\tau|$, thus 
	\begin{equation*}
		I_2=\frac{\partial \tau^{-\frac{1}{\alpha}}}{\partial\alpha} = \tau^{-\frac{1}{\alpha}}\left(\frac{\ln\tau}{\alpha^2}+\frac{1}{\alpha\tau}(-\frac{\partial\tau}{\partial\alpha})\right)>0.
	\end{equation*}
Therefore, we obtain
	\begin{equation}\label{partial3}
		\frac{\partial d_1}{\partial\alpha}>0.
	\end{equation}
	
Now we can compute partial derivatives of the integrals in \eqref{FMLS VFp}. 
	\begin{equation*}
	\begin{split}
		\frac{\partial V_{E}}{\partial\alpha} =& Ke^{-\gamma\tau}\left[\int_{d_1}^\infty\frac{\partial f_{\alpha,0}}{\partial\alpha}dm-f_{\alpha,0}(d_1)\frac{\partial d_1}{\partial\alpha}\right] + e^{x-\tau-\tau^{\frac{1}{\alpha}}d_1}f_{\alpha,0}(d_1)\frac{\partial d_1}{\partial\alpha}\\
		&-e^x\left[\int_{d_1}^\infty e^{-\tau-\tau^{\frac{1}{\alpha}}m}\left(\frac{\partial f_{\alpha,0}}{\partial\alpha}+f_{\alpha,0}(-\frac{\partial\tau}{\partial\alpha}-\frac{\partial\tau^{\frac{1}{\alpha}}}{\partial\alpha}m)\right)dm \right] \\
		=& \left[Ke^{-\gamma\tau}\int_{d_1}^\infty\frac{\partial f_{\alpha,0}}{\partial\alpha}dm-e^x\int_{d_1}^\infty e^{-\tau-\tau^{\frac{1}{\alpha}}m}\frac{\partial f_{\alpha,0}}{\partial\alpha}dm \right] \\
		& + \left[f_{\alpha,0}(d_1)\frac{\partial d_1}{\partial\alpha}\left(e^{x-\tau-\tau^{\frac{1}{\alpha}}d_1}-Ke^{-\gamma\tau}\right)\right] \\
		& -e^x\int_{d_1}^\infty f_{\alpha,0}e^{-\tau-\tau^{\frac{1}{\alpha}}m}\left(-\frac{\partial\tau}{\partial\alpha}-\frac{\partial\tau^{\frac{1}{\alpha}}}{\partial\alpha}m\right)dm \\
		=:& Q_1+ Q_2 + Q_3.
	\end{split}
	\end{equation*}
Observing that $K\geq e^x$, and $\gamma\tau=r(T-t)\to0$ as $t\to T$. Hence, from \eqref{partial1}--\eqref{partial3}, we have $Q_1<0$ and $Q_2<0$. 
Last, since $\frac{\partial\tau}{\partial\alpha}<0$ and $\frac{\partial\tau^{\frac{1}{\alpha}}}{\partial\alpha}<0$, it is obvious that $Q_3<0$ as well. 
Finally, we prove that $\frac{\partial V_{E}}{\partial\alpha}<0$, namely as $\alpha\to2$, $V_{E}$ decreasingly approaches to the European put option price under the standard Black-Scholes model for large underlying $x$, and thus verify the monotone property as observed in the numerical examples in \cite{CXZ1}.
\end{proof}

With the help of the above result and Bermudan options we can prove the second main theorem in this paper as follows.

\begin{theorem}\label{mt1}
For sufficiently large underlying the price of an American put determined under the FMLS model decreasingly converges to its Black--Scholes price as $\alpha\uparrow 2$.
\end{theorem}

\begin{proof}
The proof is via an approximation to the American option price by a sequence of Bermudan option price \eqref{Bermudan sequence}. 
By the inductive definition of Bermudan option, the price $B_N(S,t_{m-1};\alpha)$ is the maximum of $g(S)$ the the price of a European option $V_{E}(S,t;\alpha)$ with expiry $t_m$ and contract function 
$B_N(S,t_m;\alpha)$. And by the proof of Lemma \ref{lec3}, $B_N(S,0;\alpha)\to V(S,0;\alpha)$ as $N\to\infty$.

Since $V_{E}(S,t;\alpha)\to V_{BS}(S,t)$ decreasingly as $\alpha\to2$ for large $S$, $B_N(S,t;\alpha)\to B_N(S,t;2)$ decreasingly for large $S$, where $B_N(S,t;2)$ is the Bermudan option corresponding to the standard BS model. 
Then we take $N\to\infty$ to make the Bermudan option denser, and $t\to T$ to obtain $V(S,0;\alpha)\to V(S,0;2)$ decreasingly as $\alpha\to2$, for large $S$.
Note that $V(S,t;2)=V(S,t)$ is the American put option price under the standard BS model. Therefore, the theorem is proved.
\end{proof}

\section{Conclusion}\label{sec7}

This paper investigates analytic properties of American put options under the FMLS model.
A major advantage of the FMLS model over the classical Black--Scholes model is that it captures the fat tail feature (leptokurtosis) commonly observed in real world markets.
Using the closed-form analytical solutions for European options obtained in~\cite{CXZ1} and approximating an American option by a sequence of Bermudan options, 
we prove that the American put price under the FMLS model is convex with respect the underlying price, and decreasingly converges to its Black--Scholes price as the tail index $\alpha$ tends to $2$. 
The first result coincides with the corresponding property under the BS framework,
while the second result reveals that the BS model tends to underprice put options, and the pricing bias becomes larger as the tail of the return distribution becomes fatter.

\bibliographystyle{elsarticle-harv}
\bibliography{ref}

\end{document}